\documentclass[runningheads]{llncs}
\usepackage{graphicx}
\usepackage{xcolor}
\usepackage{amsmath}
\usepackage{hyperref}
\usepackage{amssymb}
\usepackage{listings}
\usepackage{cleveref}
\graphicspath{{graphics/}}
\definecolor{teal}{rgb}{0.0,.5,.5}

\DeclareMathOperator*{\va}{\mathnormal{x}}
\DeclareMathOperator*{\vb}{\mathnormal{y}}
\DeclareMathOperator*{\vc}{\mathnormal{z}}
\DeclareMathOperator*{\vab}{\mathnormal{xy}}

\newcommand{\T}{\mathcal{T}}
\newcommand{\ket}[1]{|#1\rangle}

\begin{document}
\counterwithin{lstlisting}{section}

\title{A Query-Efficient Quantum Algorithm
for Maximum Matching on General Graphs}
\titlerunning{Quantum Maximum Matching}

\author{Shelby Kimmel\inst{1} \and
R. Teal Witter\inst{2}}

\authorrunning{S. Kimmel and R.\,T. Witter}

\institute{Department of Computer Science, Middlebury College, USA
\email{skimmel@middlebury.edu} \and
Department of Computer Science and Engineering,
NYU Tandon School of Engineering, USA
\newline
\email{rtealwitter@nyu.edu}}

\maketitle

\begin{abstract}
We design quantum algorithms for maximum matching.
Working in the query model, in both adjacency matrix
and adjacency list settings, we improve on the 
best known algorithms for general graphs, matching 
previously obtained results for bipartite graphs.
In particular, for a graph with $n$ vertices and $m$ edges,
our algorithm makes $O(n^{7/4})$ queries
in the matrix model
and $O(n^{3/4}(m+n)^{1/2})$ queries
in the list model.
Our approach combines Gabow's classical maximum matching 
algorithm [Gabow, \textit{Fundamenta Informaticae}, '17] 
with the guessing tree method
of Beigi and Taghavi [Beigi and Taghavi, \textit{Quantum}, '20].
\end{abstract}

\section{Introduction}
A matching is a set of non-adjacent edges in an
undirected graph.
In the maximum matching problem, 
one tries to find the matching with the largest number of edges.
Finding the maximum matching in a graph is a problem
that is both of fundamental and practical importance.
Its practical applications range from kidney exchange 
to scheduling to characterizing chemical structures
\cite{roth_pairwise_2005,fujii_optimal_1969,may_cheminformatics_2015}.
As a fundamental problem, it has stimulated
a string of algorithmic developments, 
such as the use of blossoms and dual variables
\cite{edmonds_paths_1965},
which have been useful in the development 
of a broad range of algorithms.
Additionally, maximum matching in general
(bipartite and non-bipartite) graphs is notable for the
difficulty researchers have had in finding
a simple and correct algorithm
for this seemingly straightforward problem
\cite{micali_osqrtve_1980,gabow_weighted_2017}. 

We study maximum matching in the query setting:
We are given a graph $G$ as an adjacency matrix or 
adjacency list and the goal
is to find a maximum matching with as few
queries as possible.
A query in the matrix model takes the form,
``Do vertices $\va$ and $\vb$ share an edge?''
A query in the list model takes the form,
``What is the $i$th vertex adjacent to vertex $\va$?''

The best classical algorithms for maximum matching
solve the problem in $O(m \sqrt{n})$ time for both
bipartite and general graphs
\cite{gabow_weighted_2017,hopcroft_n52_1973,micali_osqrtve_1980,vazirani_simplification_2013}.
The query complexity of these classical algorithms
is the trivial $O(n^2)$ in the matrix model and
$O(m)$ in the list model.
In fact, using an adversarial argument,
it is easy to see that any classical algorithm must
query all pairs of vertices or all edges to
find a maximum matching in the worst case.

Using quantum computers, however, we can do better.
Lin and Lin found a quantum algorithm that solves
maximum matching on a bipartite graph in
$O(n^{7/4})$ queries in the matrix model
\cite{lin_upper_2016}.
Beigi and Taghavi created an algorithm that uses
$O(n^{3/4}\sqrt{m+n})$ queries in the list model for bipartite graphs
\cite{beigi_quantum_2020}, which in the worst case when 
$m=\Omega(n^2)$, matches the result of Lin and Lin.
Both results use the guessing tree method:
Lin and Lin introduced the method for functions 
with binary input and
Beigi and Taghavi generalized it to functions 
with non-binary input.

Our contribution is a quantum maximum matching algorithm for general graphs
that uses $O(n^{7/4})$ queries in the matrix model and 
$O(n^{3/4}\sqrt{m+n})$ in
the list model, matching the prior results for bipartite graphs.
We combine two powerful techniques to obtain our result:
Beigi and Taghavi's
guessing tree method and Gabow's
relatively simple algorithm for maximum matching
\cite{beigi_quantum_2020,gabow_weighted_2017}. 
The key technical issues in combining these two approaches are a careful 
accounting of which steps of the classical algorithm actually
require queries, slight modifications to the classical algorithm that help us bound the number of queries, and a well-chosen definition of the guessing 
scheme for the decision tree used in the guessing tree method.

The previous best known quantum algorithms for
maximum matching on general graphs ran in
trivial query complexity.
Ambainis and Špalek designed algorithms
for general maximum matching
that run in $O(n^{5/2} \log n)$ time
in the matrix model
and $O(n^2 (\sqrt{m/n} + \log n) \log n)$ time
in the list model
\cite{ambainis_quantum_2006}.
Dörn found an algorithm for general maximum matching
that runs in $O(n^2 \log^2 n)$ time in the matrix model
and $O(n\sqrt{m} \log^2 n)$ time in the list
model \cite{dorn_quantum_2009}.

While our result unifies the
cases of bipartite and general graphs,
there remains a gap between
our upper bound and the best known lower bound.
Berzina et al. and Zhang found a
lower bound for maximum matching of $O(n^{3/2})$
\cite{berzina_quantum_2004,zhang_power_2005}.
Interestingly, Zhang proved that Ambainis techniques
(one of the most useful methods for finding 
quantum lower bounds) cannot improve the
current lower bound
\cite{ambainis_quantum_2002,zhang_power_2005}.

\subsection{Graph Theory}
\label{sec:intro_graph}
Given an undirected graph $G$,
we denote by $V(G)$ the set of vertices and 
$E(G)$ be the set of edges of $G$.
Call $n = |V(G)|$ the number of vertices 
in a graph and $m = |E(G)|$
the number of edges.
We represent an edge between
vertices $\va$ and $\vb$ as $\vab $.

We denote the \textit{symmetric difference} of two graphs
$G_1$ and $G_2$ as $G_1 \oplus G_2$.
Then $V(G_1 \oplus G_2)$ is $V(G_1) \cup V(G_2)$ and $\vab \in E(G_1 \oplus G_2)$ if and only
if $\vab \in E(G_1)$ but $\vab \notin E(G_2)$
or $\vab \in E(G_2)$ but $\vab \notin E(G_1)$.
We may think of the symmetric difference as the graph equivalent
of addition modulo 2.

A \textit{matching} $M$ is a set of non-adjacent edges of $G$.
That is, if $\vab $ is in $M$,
then there is no other edge connected to $\va$ or $\vb$ in $M$.
The solid edges in \Cref{fig:graph} form a matching.
A \textit{maximum matching} on $G$ is a matching with the most edges
of any matching on $G$.
We call a vertex a \textit{free vertex} if it is not on any edge in
matching $M$, while if a vertex is not free we called it \textit{matched}. A \textit{matched edge} is in a matching while
an \textit{unmatched edge} is not.

A \textit{blossom} is a cycle of length $2k + 1$
with $k$ matched edges and $k + 1$ unmatched edges.
The edges alternate between matched and unmatched edges with the
exception of the two edges connected to the root of the blossom.
In \Cref{fig:graph}, the blossom has $2(2) + 1 = 5$ edges
and the root is the vertex in the cycle closest to the
left free vertex.

\begin{figure}[h]
\centering
\includegraphics[scale=.25]{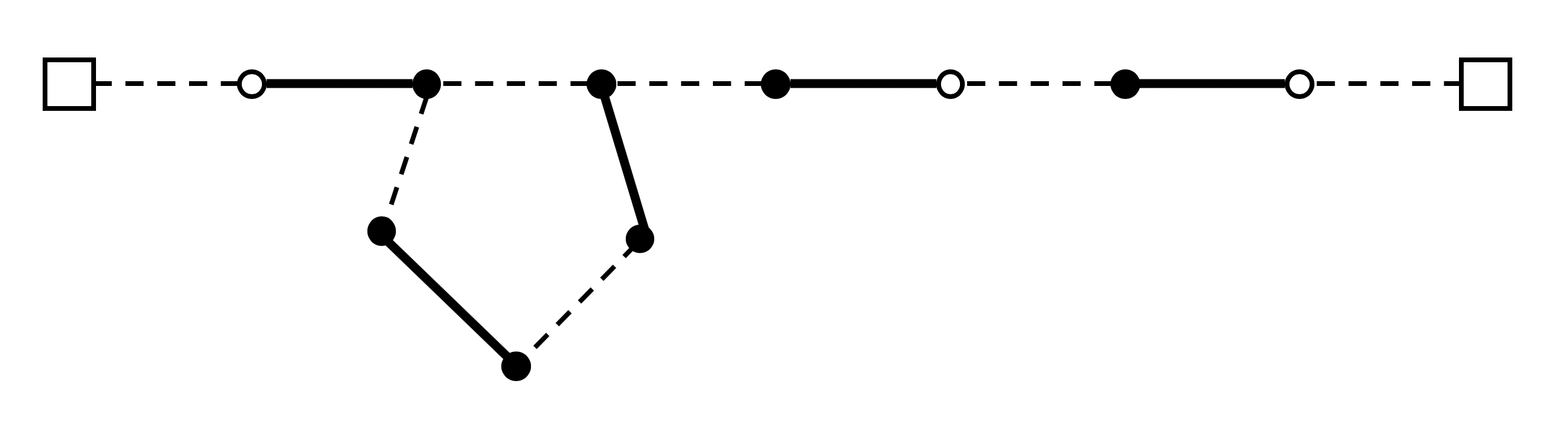}
\caption{Example graph with a matching where the solid
lines are edges in the matching and the dotted lines are
edges not in the matching but in the underlying graph.
The free vertices are squares, the outer vertices
(excluding the free vertices) are filled circles,
and the inner vertices are hollow circles.}
\label{fig:graph}
\end{figure}

An \textit{augmenting path} is a set of 
edges between two free vertices that
alternates between matched and unmatched edges. 
In \Cref{fig:graph}, the horizontal edges connecting 
the two free vertices (represented as squares) is not
an augmenting path because there are two 
consecutive unmatched edges.
A \textit{sap} (shortest augmenting path) is an augmenting
path with the fewest edges of any augmenting path.
In \Cref{fig:graph}, the augmenting path along the blossom
between the free vertices forms a sap.
We call a vertex \textit{inner} with respect to an
augmenting path if it is closer than its matched
pair (the vertex with which it shares a matched edge) to the
closest free vertex.
Here `closeness' is measured by the number of edges on the augmenting path between the vertex in question and the closest free vertex.
Inner vertices are illustrated in \Cref{fig:graph}
as hollow circles.
All other vertices---including free vertices, all vertices on a blossom,
and vertices adjacent to an edge 
equidistant between two free vertices---are \textit{outer}.
Whether a vertex is inner or outer may change as the augmenting paths grow:
An inner vertex can become outer (e.g. if it becomes part of a blossom)
but an outer vertex cannot become inner.

Notice that we can use the partial matching and sap in
\Cref{fig:graph} to get a
larger (in this case maximum) matching.
We simply take the symmetric difference of the partial
matching and augmenting path.
That is, we include every unmatched edge
(since it is in augmenting path
but not the partial matching) and remove every matched edge
(since it is in both the augmenting
path and the partial matching).
The result is a larger matching 
where each vertex with an edge in the partial matching
has an edge in the larger matching and the previously
free vertices also have matched edges.

\subsection{Query Complexity}
\label{sec:guess}


In both the list and matrix models,
we learn the edges of $G$ by
querying (i.e. evaluating at various inputs)
a function.
We assume that $G$ is a subgraph of the complete graph of $n$ vertices, labeled by elements of 
$[n]=\{0,1,\dots,n-1\}$,
where we do not know which edges of the complete graph are part of $G$ and which are
not.\footnote{One can easily extend to the case that $G$ is a subgraph of a multigraph;
we consider complete graphs only for simplicity.} Then
in the case of the adjacency matrix, we have a function $E_M:[n]\times [n]\rightarrow \{0,1\}$, where
$E_M(\va, \vb)=1$ if and only if the edge $\vab \in E(G)$.

In the case of the adjacency list, we have a function $E_L:[n]\times [n]\rightarrow [n]\cup\{\texttt{null}\}$ where
\begin{equation}
    E_L(\va,i)=
    \begin{cases}
    \vb &\textrm{ if $\vb$ is the }i\textrm{th neighbor of } \va\\
    \texttt{null}&\textrm{ if u has less than }i\textrm{  neighbors}
    \end{cases}
    \nonumber.
\end{equation}

Given access to one of these functions, the classical bounded error query
complexity of maximum matching is the number of times we must evaluate the
function in order to find a maximum matching with high probability.

In the quantum model, we are given access to unitaries called oracles that
encode the information of the functions $E_M$ and $E_L$. In the adjacency
matrix model, we have access to an oracle $O_M$ that acts on the Hilbert space
$\mathbb{C}^{n}\times \mathbb{C}^{n}\times\mathbb{C}^2$ such that for an edge
$e=\vab$, and $b\in \{0,1\}$, $O_{M}\ket{e}\ket{b}=\ket{e}\ket{b\oplus E_M(e)}$,
where addition is modulo 2.
In the adjacency list model, we have access to an oracle $O_L$ that acts on
the Hilbert space $\mathbb{C}^{n}\times \mathbb{C}^{n}\times\mathbb{C}^{n+1}$,
where for a vertex $\va$, index $i$, and $j \in[n+1]$, acts as
$O_L\ket{\va,i}\ket{j}=\ket{\va,i}\ket{j \oplus E_L(\va,i)}$, where addition is modulo
$n+1.$

Given access to one of these oracles, the quantum bounded error query
complexity of maximum matching is the number of times we must apply the oracle
(as part of a quantum algorithm) in order to find a maximum matching with high
probability.

Given a classical query algorithm,
one can create a decision tree that
describes the sequence and outcomes 
of queries that are made throughout the
algorithm. Each non-leaf vertex in the 
tree represents a query, and the outgoing 
edges from a vertex represent possible
outcomes of the query.
Sets of query outcomes may be grouped into
a single edge (provided future decisions 
made by the algorithm are independent of 
which particular query outcome within the set occurred).
Given such a decision tree, 
one can create a guessing scheme. 
A guessing scheme is a labeling of 
edges such that exactly one outgoing edge
from each vertex
is labelled as the guess. If
the outcome of a query matches the guess, 
we say that the guessing scheme
correctly guessed the outcome of that query.
Otherwise, we say it was an
incorrect guess. 

Given such a decision tree and guessing algorithm, it is possible to design a
quantum algorithm:

\begin{theorem}
[Guessing Tree \protect{\cite{beigi_quantum_2020}}]
\label{thm:guess}
    For positive integers $k$, $\ell$, and $m$,
    let $f:D_f \rightarrow [k]$ be a function with $D_f \subseteq [\ell]^m$.
    Let $\T$ be a decision tree for $f$
    with a guessing scheme and
    let $T$ be the depth of $\T$. Define
    $I$ as the maximum number of incorrect guesses
    in any path from the root to a leaf of $\T$.
    Then the bounded error quantum query complexity of evaluating 
    $f$ is upper bounded by $O(\sqrt{TI})$.
    The quantum space complexity is $O(m)$.
\end{theorem}

See Beigi and Taghavi \cite{beigi_quantum_2020} for extensive applications
of \Cref{thm:guess}.
Observe that the size of the image of the function $f$ does
not affect the query complexity or space complexity of the
quantum algorithm that evaluates it.
We use this fact to specify the maximum matching
(all $O(n)$ edges)
in the leaves of our decision tree.

\section{Result}

We use Gabow's algorithm to find
a maximum matching in graph $G$.
Gabow's algorithm runs in two phases.
(The high level pseudocode is in
\Cref{list:overview}.)
In the first phase, the algorithm finds
all the edges in $G$ that are on saps.
In the second phase, the algorithm finds
disjoint saps that are used to
augment the partial matching.
Since a maximal set of disjoint saps are found
in each iteration, there are at most
$O(\sqrt{n})$ iterations \cite{hopcroft_n52_1973}.

\begin{minipage}{\linewidth}
\begin{lstlisting}[caption={Gabow's Algorithm \cite{gabow_weighted_2017}},label=list:overview,mathescape=true, numbers=left]
$M \leftarrow \emptyset$ /* M is the current partial matching */
loop
  /* Phase 1 */
  for every pair of vertices $\va,\vb$ do
    if $\vab \in M$ then $w(\va,\vb) \leftarrow 2$ else $w(\va,\vb) \leftarrow 0$
  Listing 2/3 (matrix/list model) to find pairs of vertices on saps
  if no augmenting path is found then
    break /* $M$ has maximum cardinality */
  
  /* Phase 2 */
  Listing 4/5 (matrix/list) to create maximal set of disjoint saps $P$
  augment $M$ by the paths of $P$
\end{lstlisting}
\end{minipage}

The key idea behind the algorithm
is the use of dual variables associated with each vertex, 
and which we denote using 
a function $d:V\rightarrow \mathbb{Z}$.
Each dual variable is initialized to 1.
A pair of vertices
is tight if the sum of the dual variables
$d(\va)$ and $d(\vb)$ is $w(\va,\vb)$.
Recall from \Cref{list:overview} that
$w(\va,\vb)$ is 2 if $\vab$ is a matched edge and
0 otherwise.
Intuitively, a pair of vertices is tight only if 
their shared edge could be part of a sap
\cite{gabow_weighted_2017}.

We use Gabow's maximum matching algorithm
to construct a decision tree that finds
a maximum matching.
To apply \Cref{thm:guess} to
the decision tree, we must design a guessing scheme.
In the matrix model, 
we always guess that the edge we are querying is not present.

In the list model, when we are querying the $i^{th}$ vertex
adjacent to $\va$ (call it $\vb$),
our guess depends on the phase of the algorithm.
In the first phase, we guess that $\va$ and
$\vb$ do \textit{not} fit either of the following criteria:
\begin{itemize}
    \item $\va$ and $\vb$ are tight, $\va$ and $\vb$ are not from
    the same blossom, and $\vb$ has not yet been found
    (i.e. added to $S$, see \Cref{list:bfs_list}), or
    \item $\va$ and $\vb$ are tight, $\va$ and $\vb$ are not from
    the same blossom, and $\vb$ is outer.
\end{itemize}
In the second phase, we guess that $\va$ and $\vb$ do
\textit{not} fit either of the following criteria:
\begin{itemize}
    \item $\va$ and $\vb$ are tight,
    $\va$ and $\vb$ do not share a matched edge, and
    $\vb$ has not yet been found (i.e. added to $S'$, 
    see \Cref{list:dfs_list}), or
    \item $\va$ and $\vb$ are tight,
    $\va$ and $\vb$ do not share a matched edge, and
    $\va$ and $\vb$ form a blossom.
\end{itemize}
If our query to the list returns $\texttt{null}$,
that is, we have reached the end of a vertex's adjacency list,
we say that our guess is incorrect.

In the list model, while there might be multiple 
outcomes of a single
query that satisfy the correct guess conditions, 
we will see that the subsequent
behavior of the algorithm is the same, 
so we group all such correct outcomes into a 
single edge in our decision
tree, as described in \Cref{sec:guess}.


 Applying the above guessing scheme to Gabow's algorithm,
 we prove our main result:

\begin{theorem} \label{thm:main}
    Given a graph $G$ with $m$ edges and $n$ vertices,
    there is a bounded error quantum algorithm that finds a maximum
    matching in $O(n^{7/4})$ queries in the matrix
    model and $O(n^{3/4}\sqrt{m+n})$ queries in the
    list model.{}
\end{theorem}

In the remainder of this section,
we explain enough of Gabow's algorithm to analyze 
the performance of the quantum algorithm
and to prove \Cref{thm:main}.
However, we do not address the correctness
of Gabow's algorithm or provide sufficient
details to understand why the algorithm is correct.
Instead, we encourage interested readers
to peruse Gabow's paper \cite{gabow_weighted_2017}.

The choice to not make this paper self-contained
is intentional: including the full details of
Gabow's algorithm would double the length of this work
without adding any novel contributions.

\subsection{Breadth-First Search Subroutine}
\label{sec:bfs}

The first phase of Gabow's algorithm is a simplified search
based on Edmonds' algorithm that explores $G$ breadth-first
\cite{edmonds_paths_1965}.
The goal is to identify all the edges that are on saps.
For this purpose, the algorithm maintains
a subgraph $S$ of $G$
with the vertices and edges that have been explored.
Initially, $S$ consists of only free vertices.
As the algorithm progresses, edges and vertices are added.
We call the set of edges and vertices connected to a
free vertex a \textit{search tree}.
The algorithm terminates once two search trees become
connected i.e. there is an augmenting path from
one free vertex to another.

The algorithm also maintains a record of the blossom that
contains $\va$, denoted by $B_{\va}$.
We initially set $B_{\va}=\va$ since every vertex is
a trivial blossom and redefine $B_{\va}$ when
merging blossoms.
When all tight pairs of vertices have been checked
and no sap has been found,
the dual variables are adjusted
to find new tight pairs of vertices.
If the dual variables cannot be adjusted,
there are no augmenting paths and
the partial matching is maximum.

The execution of the simplified search based on
Edmonds' algorithm
depends on the data structure of the input graph.
In the case of the matrix model
described in \Cref{list:bfs_matrix},
we first identify vertices $\va$ and $\vb$
that fit the criteria on Line 4.
We then query the edge $\vab$ only if
$\va$ and $\vb$ satisfy either the if-statement
on Line 5 or the if-statement on Line 8.
If we reach neither Line 6 nor Line 9
then no query is made in that iteration.
If we make a query on Line 6 or Line 9
and the edge is not present,
our guess is correct.
In order to bound the number of incorrect guesses,
we bound the number of times we reach Line 7
and Line 10 which happens only if $\vab$ is present
and is in the grow, blossom, or sap case.

\begin{minipage}{\linewidth}
\begin{lstlisting}[caption={Simplified Search based on Edmonds’ Algorithm in the Matrix Model \cite{gabow_weighted_2017}},label=list:bfs_matrix,mathescape=true, numbers=left]
for every vertex $\va$ do $d(\va) \leftarrow 1$
make every free vertex outer and add to $V(S)$
loop
  if $\exists$ tight pair $\va, \vb$ with $\va$ outer, $B_{\va} \neq B_{\vb}$ then
    if $\vb \notin V(S)$ then /* grow step */
      if $\vab \in E(G)$ /* query */ then
        add $\va \vb, \vb \vb'$ to $S$ where $\vb \vb' \in M$
    else if $\vb$ is outer then
      if $\vab \in E(G)$ /* query */ then
        if $\va$ and $\vb$ in the same search tree then
          /* blossom step */
          merge all blossoms in fundamental cycle of $\vab$
        else /* $\vab$ forms a sap */
          return /* continue in Listing 1 */
  else
    dual adjustment step
    /* no queries are made, see Gabow Figure 2 for details */
\end{lstlisting}
\end{minipage}

In the case of the list model described in \Cref{list:bfs_list},
we query from an outer vertex $\va$ and find some 
adjacent vertex $\vb$.
If $\va$ and $\vb$ are not tight,
$\va$ and $\vb$ are not from the same blossom
or neither of the criteria on Lines 9 and 11 apply,
then our guess is correct.
We bound the number of incorrect guesses by the number of times we reach Lines 7,
10, and 12, which happens only if
we have reached the end of $\va$'s neighbors or $\va$ and $\vb$
are in the grow, blossom, or sap case.

Observe that we can group the correct
guesses in the list model into a single edge
in the decision tree because the algorithm's
behavior is the same in every case:
continue to query neighbors of $\va$.

\begin{lemma}\label{lemma:bfs}
    The simplified search of Edmonds' algorithm
    makes at most $O(n)$ incorrect guesses
    in a single call.
\end{lemma}

\begin{proof}
    As discussed above,
    in both the matrix and list models,
    a guess is incorrect only if
    we are in the grow, blossom, or sap case
    (or in the list model at the end of a list). 
    Therefore we bound the number of incorrect guesses
    by the number of times we can reach each case.
    In the grow case where $\vb \notin S$,
    we add both $\vb$ and $\vb'$ to $S$, where $\vb\vb'$
    is in the current partial matching $M$.
    Since this case only occurs when a vertex $\vb$ is not in $S$,
    and there are at most $n$ vertices in the graph,
    this case can trigger at most $n$ incorrect guesses.

    In the blossom case where $\va$ and $\vb$ are in the
    same search tree, we have merged at least two
    blossoms.
    Each vertex is initially a blossom
    so we start with a total of $n$ blossoms.
    Each time we merge two or more blossoms,
    we reduce the number of blossoms by 
    at least one.
    Therefore we can merge blossoms at most
    $n$ times, and so we can only make $n$
    incorrect guesses in this case.

    In the case where $\vab$ completes a sap,
    we halt the algorithm and so this may
    happen at most once per call.
    In the list model, we can reach the end
    of a list at most $n$ times so the number
    of incorrect guesses due to \texttt{null} outcomes is bounded by $n$.
\end{proof}

\begin{minipage}{\linewidth}
\begin{lstlisting}[caption={Simplified Search based on Edmonds’ Algorithm in the List Model},label=list:bfs_list,mathescape=true, numbers=left]
for every vertex $\va$ do $d(\va) \leftarrow 1$
make every free vertex outer and add to $V(S)$
loop
  for every outer vertex $\va$ do
    for every vertex $\vb$ adjacent to $\va$ do
      if $\vb$ is null then /* end of list */
        break /* go to next $\va$ */
      else if $\va$ and $\vb$ are tight and $B_{\va} \neq B_{\vb}$ then
        if $\vb \notin V(S)$ then /* grow step */
          add $\va \vb, \vb \vb'$ to $S$ where $\vb \vb' \in M$
        else if $\vb$ is outer then
          if $\va$ and $\vb$ in the same search tree then
            /* blossom step */
            merge all blossoms in fundamental cycle of $\vab$
          else /* $\vab$ forms a sap */
            return /* continue in Listing 1 */
  dual adjustment step
  /* no queries are made, see Gabow Fig. 2 for details */
\end{lstlisting}
\end{minipage}

\subsection{Depth-First Search Subroutine}
\label{sec:dfs}

In the second phase of the 
algorithm---the path-preserving depth-first
 search---we identify disjoint saps.
We define a 
subgraph $H$ of the complete graph
which we initialize
with the edges between every pair of
tight vertices in $S$.
(While many edges in $H$ were queried
in the breadth-first subroutine, not all were; in particular,
most edges between search trees
have not yet been queried.)
The algorithm explores $H$ from each free
vertex in order to find another free vertex.

\begin{minipage}{\linewidth}
\begin{lstlisting}[caption={Path-Preserving Depth-First Search in the Matrix Model \cite{gabow_weighted_2017}},label=list:dfs_matrix,mathescape=true, numbers=left]
initialize $P$ to an empty set
for each free vertex $f$ do
  if $f \notin V(P)$ then
    initialize $S'$ to an empty graph
    add $f$ to $S'$ as the root of a new search tree
    $find\_ap(f)$

procedure $find\_ap(\va:$ /* $\va$ is an outer vertex */
  for each edge $\vab \in E(H) \setminus M$ do
    if $\vb \notin V(S')$ then
      if $\vab \in E(G)$ /* query */ then
        if $\vb$ is free then /* $\vb$ completes a sap */
          add $\vab$ to $S'$ and sap to $P$
          terminate all current recursive calls to $find\_ap$
          remove all edges of sap from $H$
          recursively remove all dangling edges from $H$
        else /* grow step */
          add $\vab, \vb \vb'$ to $S'$ where $\vb \vb' \in M$
          $find\_ap(\vb')$
          /* accessible only if $\vb'$ is not on a sap */
          remove $\vb$ and $\vb'$ from $H$
      else
        remove $\vab$ from $H$
        recursively remove all dangling edges from $H$
    else if blossom found then
      if $\vab \in E(G)$ /* query */ then
        blossom procedure /* see Gabow Fig. 4 for details */
        /* calls $find\_ap(\va)$ from each vertex $\va$ in blossom */
      else
        remove $\vab$ from $H$
        recursively remove all dangling edges from $H$
\end{lstlisting}
\end{minipage}

While $H$ contains edges on saps, one edge can be on
more than one sap. 
This is a problem, as we need disjoint saps in order 
to augment the partial matching.
To account for this, using recursive calls, 
the depth-first search explores $H$ from a single free
vertex and forms
a new subgraph $S'$ of visited vertices along the way.
Once another free vertex is found from the starting free vertex,
the algorithm processes the sap and terminates all
current calls, disallowing edges of the present sap from
being used in future saps and reinitializing $S'$.
Then another call is made from a new free vertex.
If the algorithm identifies a vertex on a blossom that
has already been explored, 
new recursive calls are initiated from each vertex
on the blossom.

We maintain the property that all
edges in $H$ are on as yet unidentified
saps by deleting
edges and vertices in several cases:
When we find a sap, we remove all the edges
and vertices along it.
Thus no remaining sap in $H$ can share
an edge with one that was already found.
When we query an edge that is not present,
we remove it from $H$. 
When the recursive call does not find a
sap containing vertex $\va$, we remove $\va$ and its adjacent edges.
After deletions, some \textit{dangling} edges
may remain in $H$.
A dangling edge has an adjacent vertex with degree one
(as a result of a deletion) that is not a free vertex.
We remove dangling edges from $H$ by recursively deleting
the edge and adjacent vertex with degree one in addition
to resulting dangling edges.

\begin{minipage}{\linewidth}
\begin{lstlisting}[caption={Path-Preserving Depth-First Search in the List Model},label=list:dfs_list,mathescape=true, numbers=left]
initialize $P$ to an empty set
for each free vertex $f$ do
  if $f \notin V(P)$ then
    initialize $S'$ to an empty graph
    add $f$ to $S'$ as the root of a new search tree
    $find\_ap(f)$

procedure $find\_ap(\va):$ /* $\va$ is an outer vertex */
  for every vertex $\vb$ adjacent to $\va$ do
    if $\vb$ is null then /* end of list */
      break /* go to origin of current call to $find\_ap$ */
    else if $\vab \in E(H) \setminus M$ then
      if $\vb \notin V(S')$ then
        if $\vb$ is free then /* $\vb$ completes a sap */
          add $\vab$ to $S'$ and sap to $P$
          terminate all current recursive calls to $find\_ap$
          remove all edges of sap from $H$
          recursively remove all dangling edges from $H$
        else /* grow step */
          add $\va \vb, \vb \vb'$ to $S'$ where $\vb \vb' \in M$
          $find\_ap(\vb')$
          /* accessible only if $\vb'$ is not on a sap */
          remove $\vb$ and $\vb'$ from $H$
      else if blossom found then
        blossom procedure /* see Gabow Figure 4 for details */
        /* calls $find\_ap(\va)$ from each vertex $\va$ in blossom */
\end{lstlisting}
\end{minipage}

Gabow's original version of the path-preserving
depth-first search does not need to maintain the property
that all edges in $H$ are on as yet unidentified saps
since other edges can be weeded out through the course
of the algorithm.
Since our goal is to bound costly ``incorrect'' queries,
we cannot afford to wait to remove
these edges and must preemptively do so. We need to ensure that 
this modification does not affect the correctness of the algorithm, but
it is easy to see that the edges we remove
from $H$ (described in the previous paragraph)
can not be part of any as yet undiscovered disjoint saps. 
Since the purpose of this subroutine is to discover a set of disjoint saps,
this modification does not affect the correctness of this phase. This change might
affect the runtime, but as we are concerned with query complexity rather
than time complexity, we will not further analyze the runtime consequences.

The 
path-preserving
depth-first search depends
on the data structure of the input graph.
In the case of the matrix model described in \Cref{list:dfs_matrix},
we identify vertices $\va$ and $\vb$
that fit the criteria on Line 9 and either
Line 10 or Line 25.
We then query the edge $\vab$ on Line 11
or Line 26.
If the edge is not present, our guess is correct.
In order to bound the number of incorrect guesses,
we bound the number of times we reach Line 12
and Line 27,
which happens only if $\vab$ is present
and completes a sap, triggers a grow step, or
forms a blossom.

In the case of the list model described in \Cref{list:dfs_list},
we query from outer vertex $\va$ and
find some adjacent vertex $\vb$.
If $\va$ and $\vb$ are not tight,
$\va$ and $\vb$ share a matched edge, or
neither of the criteria on Lines 13 and 24 apply,
then our guess is correct. 
While there might be multiple query outcomes 
that count as correct,
the algorithm behaves the same in each case: 
continue to query the next neighbor of $\va$.
In order to bound the number of incorrect guesses,
we bound the number of times we reach Lines 11,
13, and 25,
which happens only if we have reached the end
of $\va$'s neighbors or $\va$ and $\vb$ complete
a sap, trigger a grow step, or form a blossom.

\begin{lemma}\label{lemma:dfs}
    The path-preserving depth-first search makes at most
    $O(n)$ incorrect guesses in a single call.
\end{lemma}

\begin{proof}
    In both the matrix and list models, a guess is incorrect
    only if we are in the sap, grow, or blossom case.
    Therefore we bound the number of incorrect guesses
    by the number of times we can reach each case.
    If $\vb$ is a free vertex, we
    have found a sap and immediately remove
    $\va$ and $\vb$ from $H$
    since they lie on a sap we have found.
    Thus we can bound the number of incorrect guesses
    in this case by the number of free vertices which
    is in turn bounded by $n$.

    If $\vb$ is not a free vertex,
    $\vb$ may either be on a sap or not.
    Note that since $\vab$ is tight, it could
    be on a sap but if another edge further on the 
    potential sap is not present or the potential sap
    overlaps with a sap already in $P$ we say that $\vb$
    is not on a sap.

    If $\vb$ is not a free vertex and is on a sap,
    we remove $\va$ and $\vb$ from $H$
    once the sap is found.
    Observe that there is a one-to-one correspondence
    between the edge $\vab$ and the vertex $\vb$.
    That is, since $\vb$ is now in $S'$, we will not
    process another edge $\vc\vb$ for some vertex $\vc$.
    It follows that the number of incorrect guesses in this
    case is bounded by the number of vertices $n$.

    If $\vb$ is not a free vertex and is not on a sap,
    we will return from the call
    and remove $\vb$ and $\vb'$ from $H$ 
    (see Line 21 in \Cref{list:dfs_matrix},
    Line 23 in \Cref{list:dfs_list}).
    We can safely remove these vertices because $\vb'$
    is not on a sap and for $\vb$ to be on a sap,
    there would be two consecutive unmatched
    edges which is a contradiction.
    Then the number of incorrect guesses in this case
    is bounded by the number of vertices we can remove which
    is $n$.

    If $\va$ and $\vb$ form a blossom then we can
    bound the number of incorrect guesses
    by the number of times blossoms can be
    merged which is in turn bounded by $n$,
    the number of blossoms initially present.
    In the list model, we can reach the end
    of a list at most $n$ times so the number
    of incorrect guesses due to \texttt{null} outcomes
    is bounded by $n$.
\end{proof}

We now combine the two lemmas to prove our main result.

\begin{proof}[of \Cref{thm:main}]
    The guessing scheme is described above the statement
    of \Cref{thm:main}.
    We create a decision tree using \Cref{list:overview}.
    The depth of the decision tree is the total number
    of queries we would need to make to learn the graph $G$.
    In the matrix model, this is $n^2$.
    In the list model, this is $m+n$ because we need to check
    each vertex and all the edges in its adjacency list.
    We can ensure this bound by keeping a classical
    record of our queries and query outcomes and,
    before querying the oracle, 
    checking whether we have made
    this query before.
    By \Cref{lemma:bfs}, \Cref{lemma:dfs}, and
    the $O(\sqrt{n})$ bound on the number of iterations,
    the number of incorrect guesses is bounded by
    $O(n\sqrt{n})$.
    Then \Cref{thm:main} follows from \Cref{thm:guess}.
\end{proof}

\section{Conclusion}

We used a classical maximum matching algorithm
and the guessing tree method
to give a $O(n^{7/4})$ query bound in the matrix
model and $O(n^{3/4}\sqrt{m+n})$ query bound in
the list model for maximum matching on
quantum computers and general graphs.
Our result narrows the gap between the previous
trivial upper bounds of $O(n^2)$
and $O(m)$ and the quantum query complexity
lower bound of $O(n^{3/2})$.
An important open problem is to determine whether 
this algorithm is optimal.
Progress on this question could be made by 
improving the lower bound, perhaps using 
the general adversary bound \cite{hoyer_negative_2007}.

Another open problem is to bound the time
complexity of the guessing tree method. 
Such a result would then allow us 
to compare the maximum matching algorithm described
in this paper to existing quantum maximum matching algorithms 
that aim to minimize time complexity rather than query complexity. 
The time complexity of implementing the guessing tree
method is currently unknown.
The guessing tree algorithm is based on
the dual adversary bound \cite{beigi_quantum_2020}, and 
the quantum algorithm that results is an alternating sequence
of input-dependent and input-independent unitaries,
at least in the binary case
\cite{reichardt_span_2009,lee_quantum_2011}.
While the input-dependent unitary is 
simply the oracle and may
be applied in constant time,
the time complexity of the input-independent unitary
depends on finding an efficient
implementation of a quantum walk on the decision tree.
The guessing tree algorithm is similar to the 
\textit{st}-connectivity span program
algorithm, for which a relationship between query and time complexity is known
\cite{jeffery_quantum_2017}.
The scaling between time and query complexity in that algorithm depends
on the time complexity of implementing a quantum walk on the decision tree and on the 
spectral gap of the normalized Laplacian of the decision tree. It would be interesting
if a similar relationship holds for the guessing tree algorithm, and if so, how it 
applies to the specific case of maximum matching.

\bibliography{references}

\end{document}